\documentclass[a4paper, english, numberwithinsect]{lipics}

\theoremstyle{plain}
\newtheorem{problem}[theorem]{Problem}
\newenvironment{parameterizedproblem}%
{%
  \leavevmode\nobreak\par
  \begin{list}%
    {}%
    {%
      \def\labelstyle{\itshape}
      \setlength{\topsep}{0pt}%
      \settowidth{\labelwidth}{\labelstyle Parameter:}%
      \setlength{\leftmargin}{\labelwidth}%
      \addtolength{\leftmargin}{\labelsep}%
      \setlength{\itemsep}{0pt}%
      \setlength{\parsep}{0pt}%
    }%
      \def\instance{\item[\labelstyle Instance:]}%
      \def\parameter{\item[\labelstyle Parameter:]}%
      \def\question{\item[\labelstyle Question:]}%
    }%
    {%
  \end{list}%
}

\newcommand\Class[1]{%
  \mathchoice%
  {\text{\normalfont\fontsize{9pt}{10pt}\selectfont$\mathrm{#1}$}}%
  {\text{\normalfont\fontsize{9pt}{10pt}\selectfont$\mathrm{#1}$}}%
  {\text{\normalfont$\mathrm{#1}$}}%
  {\text{\normalfont$\mathrm{#1}$}}%
}

\newcommand\Para{\mathrm{para\text-}}

\newcommand{\Lang}[1]{%
  \ifmmode{%
    \text{\normalfont\textsc{#1}}%
  }%
  \else
  {\normalfont\textsc{#1}}%
  \fi}

\newcommand\PLang[1][]{p\def\test{#1}\ifx\test\stockhustantauempty\else_{\mathrm{#1}}\fi\text-\penalty15\Lang}

\def\stockhustantauempty{}

\newcommand\PLangText[1][]{$p\def\test{#1}\ifx\test\empty\else_{\mathrm{#1}}\penalty15\fi$-\penalty15\hskip0pt\textsc}

\newenvironment{claimagain}[1]{\par\medskip\noindent\textbf{\sffamily #1.}\ \ \bgroup\itshape\ignorespaces}{\egroup}
\let\origlemma\lemma
\def\lemma#1#2#3\end{\origlemma%
  \expandafter\gdef\csname savedlemma#2\endcsname{#3}%
  \label{#2}#3\end}
\def\reclaimlemma#1{\claimagain{Claim of Lemma~\ref{#1}}\expandafter\expandafter\expandafter\ignorespaces\csname savedlemma#1\endcsname\endclaimagain\ignorespaces}
\let\origtheorem\theorem
\def\theorem#1#2#3\end{\origtheorem%
  \expandafter\gdef\csname savedtheorem#2\endcsname{#3}%
  \label{#2}#3\end}
\def\reclaimtheorem#1{\claimagain{Claim of Theorem~\ref{#1}}\expandafter\expandafter\expandafter\ignorespaces\csname savedtheorem#1\endcsname\endclaimagain\ignorespaces}
\let\origcorollary\corollary
\def\corollary#1#2#3\end{\origcorollary%
  \expandafter\gdef\csname savedcorollary#2\endcsname{#3}%
  \label{#2}#3\end}
\def\reclaimcorollary#1{\claimagain{Claim of Corollary~\ref{#1}}\expandafter\expandafter\expandafter\ignorespaces\csname savedcorollary#1\endcsname\endclaimagain\ignorespaces}

\title{Fast Parallel Fixed-Parameter Algorithms via Color Coding}

\author[1,2]{Max Bannach}
\author[1]{Christoph Stockhusen}
\author[1]{Till Tantau}
\affil[1]{%
  Institute for Theoretical Computer Science\\
  Universit\"at zu L\"ubeck\\
  L\"ubeck, Germany\\
  \texttt{\{bannach,stockhus,tantau\}@tcs.uni-luebeck.de}
}
\affil[2]{%
  Graduate School for Computing in Medicine and Life Sciences\\
  Universit\"at zu L\"ubeck
}

\authorrunning{M. Bannach, C. Stockhusen, and T. Tantau}
\Copyright{M. Bannach, C. Stockhusen, and T. Tantau}


\serieslogo{}
\volumeinfo
 {Billy Editor and Bill Editors}
 {2}
 {Conference title on which this volume is based on}
 {1}
 {1}
 {1}
\EventShortName{}


\begin{document}

\maketitle

\begin{abstract}
  Fixed-parameter algorithms have been successfully applied to solve
  numerous difficult problems within acceptable time bounds on large
  inputs. However, most fixed-parameter algorithms are inherently
  \emph{sequential} and, thus, make no use of the parallel hardware present
  in modern computers. We show that parallel fixed-parameter
  algorithms do not only exist for numerous parameterized problems 
  from the literature -- including vertex cover, packing problems,
  cluster editing, cutting vertices, finding embeddings, or finding
  matchings -- but that there are parallel algorithms working
  in \emph{constant} time or at least in time \emph{depending only on
    the parameter} (and not on the size of the input) for these
  problems. Phrased in terms of complexity classes, we place numerous
  natural parameterized 
  problems in parameterized versions of AC$^0$. On a more technical
  level, we show how the \emph{color coding} method can be implemented
  in constant time and apply it to embedding problems for graphs of
  bounded tree-width or tree-depth and to model checking first-order
  formulas in graphs of bounded degree.
\end{abstract}

\section{Introduction}

The classical objective of parameterized complexity theory is to
determine for a parameterized problem whether it can be solved by an
algorithm running in time $f(k)\cdot n^c$, where $f$ is some
function, $k$ is a parameter, $n$ is the input length, and $c$
is some constant. Such algorithms are nowadays routinely used to solve 
large instances for $\Class{NP}$- or even $\Class{PSPACE}$-hard
problems within acceptable amounts of time. Nevertheless,
``acceptable'' is not the same as ``small'' and one would like to
further reduce the runtime 
by using multiple cores to speed up the computation. For this, one needs
\emph{parallel} fixed-parameter algorithms, but most fixed-parameter
algorithms have been devised with a sequential computation model in
mind. Indeed, the most important tool of parameterized complexity
theory, namely kernelization, is inherently sequential: 
It asks us to repeatedly apply rules to an input, each time modifying
the input slightly and making it a little smaller, until the input's
size only depends on the parameter.  There is no straightforward way
of parallelizing such algorithms since later modifications strongly
depend on what happened earlier, forcing us to apply the typically
very large number of kernelization steps in a sequential manner. 

\subparagraph*{Our Contributions.}
The purpose of the present paper is to show that not only do parallel
fixed-parameter algorithms exist for many natural, well-studied
problems from the literature; for certain problems there are even
parallel algorithms that require only \emph{constant} time in a
concurrent-read, concurrent-write
\textsc{pram} model (so the runtime is totally independent of the
input) or at least time \emph{depending only on the parameter} (so the
length of the input is irrelevant). In all cases,
the \emph{work} done by the algorithms is still $f(k) \cdot n^c$, that
is, the same as the time bound for sequential fixed-parameter
algorithms.\footnote{The \emph{work} done by a parallel algorithm is the total
number of computational steps made by all computational units during
a computation. Since ``all work needs to be done,'' in practice the
runtime of a parallel algorithm is its work divided by the number of
available cores. In particular, the \emph{work} done by
a \emph{parallel algorithm} should not exceed the \emph{runtime} of a
\emph{sequential algorithm} for the same problem. In our 
case, this means that in order to compete with sequential algorithms
running in ``\textsc{fpt} time,'' our parallel algorithm must not only
be fast, but may only do ``\textsc{fpt} work.''} Phrased more
formally, our objective is to identify parameterized problems that
lie in the complexity classes $\Para\Class{AC}_{O(1)}$ and
$\Para\Class{AC}_{f(k)}$ (formal definitions will be given later).

In order to tackle the parallel parameterized complexity of natural
problems like the vertex cover problem, we introduce three technical
tools. The first and foremost is \emph{color coding}: all of our
proofs employ this 
technique at least indirectly and we show that the \emph{universal
  coloring families} that lie at the heart of the technique can be
computed in constant time. Second, numerous natural ``packing
problems'' are special cases of the following embedding problem: Given
graphs $H$ and $G$, find a (not necessarily induced) subgraph of $G$
that is isomorphic to~$H$. We give new bounds on the complexity of
this problem when $H$ 
has bounded tree-width or bounded tree-depth; and these bounds later
translate directly to bounds on different packing problems. Third, we
translate an algorithmic meta-theorem of Flum and
Grohe~\cite{FlumG2002} to the parallel world: We show that model
checking first-order properties of graphs can be done in parallel in
time depending only on the parameters (actually, only on the locality
rank of the formula), where the parameters are the to-be-checked
formula and the degree of the graph. 

We then apply the tools to a wide variety of natural graph problems,
namely \emph{packing problems,} \emph{covering problems,}
\emph{clustering   problems,} and \emph{separation problems.}
For \emph{packing problems} the objective is to determine whether a
given graph $G$ contains $k$ vertex-disjoint copies of some fixed
graph~$H$ like, say, a triangle. Even for triangles, this problem is already
$\Class{NP}$-complete, but when $k$ is considered as a
parameter, the triangle packing problem lies in
$\Class{FPT}$~\cite{FellowsHRST2004}. We show that there is a constant-time,
\textsc{fpt}-work algorithm for triangle packing -- and indeed for
packing any graph of fixed size. The \emph{covering problems} we study
include the vertex cover problem and its partial version. We present a
constant parallel time algorithm for the first problem and an
algorithm for the second needing time depending only 
on the parameter. These results nicely reflect on a theoretical basis
the ``empirical'' observation that $\PLang{vertex-cover}$ is one of
the ``easiest'' parameterized problems and that the partial version is
a bit harder to solve. For \emph{clustering
  problems,} also known as  \emph{cluster editing problems,} the
objective is to transform a graph by adding or deleting few edges into
a collection of ``clusters'' -- which are just cliques in the simplest
case. We present a constant time, \textsc{fpt}-work algorithm for
cluster editing. For \emph{graph separation problems} the
objective is to ``cut away'' a special part of a graph using few
vertices. We show that certain versions of these problems can be
solved by a parallel fixed-parameter algorithm in time depending only
on the parameter and \textsc{fpt} work (while other versions are known
to be $\Class W[1]$-hard).

\subparagraph*{Related Work.}

There is a growing body of literature reporting on the practicalities
of implementing fixed-parameter algorithms in
parallel~\cite{AbuKhzamLSS2006}. 
In contrast, there are only few
results addressing parallel fixed-parameter tractability on a
theoretical level (as we do in the present paper), see for instance
Cesati and Di Ianni~\cite{Cesati:1998fe}. 
Since it is well-known from
classical complexity theory that problems solvable in logarithmic
space can be parallelized well, previous research on
\emph{parameterized logarithmic space} contributes to our
understanding of which parameterized problems can be parallelized in
principle. This research was started by Cai, Chen, Downey, and
Fellows~\cite{CaiCDF1997}. First (quite technical) complete problems
for parameterized logarithmic space where later introduced by Chen,
Flum, and Grohe~\cite{ChenFG2003}, and by Flum and
Grohe~\cite{FlumG2002}. A more structural study of parameterized space
and circuit classes (which addresses parallelization more directly)
was later made by Elberfeld and the last two
authors~\cite{ElberfeldST2014}.
Parameterized Circuit Complexity was
also studied by Downey et al.~with respect to the Weft
Hierarchy~\cite{Downey:1998oq}. 
Recently,
Chen and M\"uller~\cite{ChenM2014B} connected color coding and
parameterized space in an algorithm for finding embedding of bounded
tree-depth graphs in parameterized logarithmic space (a result which
we strengthen considerably in Corollary~\ref{corollary:subgraphAC}). 

The first use of the color coding technique can be traced back to
Alon, Yuster, and Zwick~\cite{AlonYZ1994}. They used the technique to
provide an $\Class{FPT}$-algorithm that decides whether there is an
embedding of a graph~$H$ of bounded tree-\emph{width} into another
graph $G$, where $H$ is the parameter.

\subparagraph*{Organization of This Paper.}

In Section~\ref{section:preliminaries} we give formal definitions of
the classes of problems solvable by parallel fixed-parameter
algorithms. While most of our definitions and classes are standard,
the class of problems solvable in ``time depending on the parameter
and \textsc{fpt} work'' seems to be new. In
Section~\ref{section:tools} we introduce our three
technical tools -- color coding, embeddings, and model checking -- and
prove the results mentioned earlier. In
Section~\ref{section:applications} we study the complexity of the
natural parameterized graph problems and establish new upper bounds on
their complexities.
Due to lack of space, proofs have been moved to the appendix; we give
proof sketches for some of them in the main text.  

\section{Classes of Fixed-Parameter Parallelism}
\label{section:preliminaries}

For our definition of parallel fixed-parameter tractability, we mostly
use the standard terminology of parameterized complexity theory, see
for instance~\cite{FlumG2006}: A \emph{parameterized problem} is a
tuple $(Q,\kappa)$ of a language $Q\subseteq\Sigma^*$ over an
alphabet~$\Sigma$ and a parameterization
$\kappa\colon\Sigma^*\to\mathbb{N}$ that maps instances to parameter
values.  
In the classical definition, Downey and Fellows~\cite{DowneyF1999}
require the parameterization to be computable, while Flum and
Grohe~\cite{FlumG2006} require it to be computable in polynomial
time. Elberfeld and the last two authors require it to be computable in
logarithmic space~\cite{ElberfeldST2014} and mention that it would be better if
the parameterization is first-order computable ($\Class{FO}$-computable) or,
equivalently, to be computable by logarithmic-time-uniform constant
depth circuits~\cite{Vollmer1999}. Since we will only deal with
parameterized circuit classes that lie within parameterized
logarithmic space, we will require all
parameterizations to be $\Class{FO}$-computable.
We denote parameterized problems with a leading ``$p$-'' as
in $\PLang{vertex-cover}$ and, when the parameter may be unclear,
add it as an index as in $\PLang[\mathit{|H|}]{emb}$.

A parameterized problem is \emph{fixed-parameter tractable} if it
can be decided in time $f(\kappa(x))\cdot |x|^{c}$ for any
input~$x$, where $f$ is some computable function and $c$ a constant. An equivalent definition is
that there exists a set $R\in\Class P$, where $\Class P$ denotes the
class of languages decidable in 
polynomial time, such that $x \in Q$ iff $\bigl(x,1^f(\kappa(x))\bigr)
\in R$. The first  definition of fixed-parameter tractability gave
rise to the class name $\Class{FPT}$ in the literature, while the
second definition gives rise to the name  
$\Para\Class P$ for the same class. The advantage of the second
definition is that we can replace the class $\Class P$ in the
definition by arbitrary complexity classes and arrive at classes like 
\emph{parameterized logarithmic space,} $\Para\Class L$, or 
\emph{parameterized constant depth circuits},
$\Para\Class{AC}^0$. These parameterized classes inherit their
inclusion structure from the classical classes, so we have
\begin{align*}
  \Para\Class{AC}^0 \subsetneq \Para\Class{TC}^0
  \subseteq \Para\Class{NC}^1 \subseteq \Para\Class{L}
  \subseteq \Para\Class{NL} \subseteq \Para\Class{AC}^1
  \subseteq \Para\Class P.
\end{align*}

It is not quite obvious, but the class $\Para\Class{AC}^0$ already
captures  one of the types of algorithms mentioned in the
introduction, namely ``constant time, \textsc{fpt}-work,'' while none
of the above classes seems to capture ``parameter time,
\textsc{fpt}-work.'' For this reason and in 
order to explicitly spell out what $\Para\Class{AC}^0$ contains, we
provide a new definition:\footnote{The definition can trivially be
  adjusted to use TC-circuits or NC-circuits, but we will not need them.}

\begin{definition}[Classes of Parallel Fixed-Parameter Tractability]
  Let $d\colon\mathbb{N}^2\rightarrow\mathbb{N}$ be a \emph{depth bounding
    function} and $w\colon\mathbb{N}^2\rightarrow\mathbb{N}$ be a
  \emph{width bounding function} which both map each pair of an input
  length and a parameter to a number. We define
  $\Para\Class{AC}[d,w]$ as the class of parameterized problems
  $(Q,\kappa)$ for which there exists a
  \textsc{dlogtime}-uniform\footnote{We use
    \textsc{dlogtime}-uniform families since they are equivalent to
    first-order definable families and constitute one of the strongest
    forms of uniformity~\cite{Barrington:1990yg}.} family $(C_{n,k})_{n,k\in\mathbb N}$ of $\Class{AC}$-circuits
  (only \textsc{not}-, \textsc{and}-, and \textsc{or}-gates 
  are allowed, \textsc{and}- and \textsc{or}- gates may have
  unbounded fan-in) such that: 
  \begin{enumerate}
      \item For all $x \in \Sigma^*$, the circuit $C_{|x|,\kappa(x)}$
    evaluates to $1$ on input $x$ if, and only if, $x \in Q$.
      \item The depth of each $C_{n,k}$ is at most $d(n,k)$.
      \item The size of each $C_{n,k}$ is at most $w(n,k)$.
  \end{enumerate}
\end{definition}

In the present paper we
exclusively study parallel algorithms with ``\textsc{fpt}-work'' and are
therefore only interested in the case where $w$ is member of the
family $W$ of functions of the form $f(k)\cdot n^c$ for a computable
function $f$ and a constant $c$.
We introduce for arbitrary families $D$ of functions
$d\colon\mathbb{N}^2\rightarrow\mathbb{N}$ the
abbreviation $\Para\Class{AC}_D$ for $\bigcup_{d\in D, w\in
  W}\Para\Class{AC}[d,w]$.
For constant depth bounding functions the resulting class
$\Para\Class{AC}_{O(1)}$ is the same as the class $\Para\Class{AC}^0$.\footnote{Since 
  the designation $\Para\text{AC}^0$ has been used in previous
  publications and is a bit shorter, we will use it in the following.}
For arbitrary $i>0$ we obtain 
\(
  \Para\Class{AC}_{\{\,
    f(k)+c\cdot\log^i n\,\mid\,
    f\colon\mathbb{N}\rightarrow\mathbb{N}\wedge c\in\mathbb{N}
   \,\}}
  =\Para\Class{AC}^i
\)(in slight abuse of notation we will write such classes simply as
$\Para\Class{AC}_{f(k)+O(\log^i n)}$).

When the depth bounding function just depends on the 
parameter, so $d(n,k) = f(k)$, we get a new
class $\Para\Class{AC}_{f(k)}$ that we abbreviate with
$\Para\Class{AC}^{0\uparrow}$. This class does not seem to arise from
substituting some classical class for $\Class P$ in the definition of
$\Para\Class P$. 
In particular, this class seems to be incomparable
with all classes between $\Para\Class{TC}^0$ and $\Para\Class{NL}$. It is, however, clearly contained in
$\Para\Class{AC}^1$, and is strictly more powerful then
$\Para\Class{AC}^0$ as we will see later. 
This class captures the problems solvable in
``parameter time, \textsc{fpt}-work'' and we have
\[
  \Para\Class{AC}^{0}\subsetneq  \Para\Class{AC}^{0\uparrow}\subseteq  \Para\Class{AC}^{1}.
\]
 Let us define for arbitrary $i\geq 0$ the class
 $\Para\Class{AC}^{i\uparrow}$ as $\Para\Class{AC}_{f(k)\cdot O(\log^i
   n)}$. Notice that we have by definition the inclusion structure 
 \(
   \Para\Class{AC}^i\subseteq\Para\Class{AC}^{i\uparrow}\subseteq\Para\Class{AC}^{i+\epsilon}.
 \)

\section{Technical Tools}
\label{section:tools}

\subsection{Color Coding in Constant Parallel Time}

The idea of color coding is best understood by a concrete application,
for instance to the well-known matching problem: Given an
undirected graph~$G$ and a number~$k$, does $G$ contain $k$ edges such
that no two of them share any endpoints? Directly solving this problem
is not easy since the known polynomial-time algorithms for it are rather
involved. Consider, however, what 
happens when we randomly color the graph with $k$ colors and then
check whether \emph{the vertices of each color class contain at least one
  edge}. Clearly, if this is the case, there is a matching of size~$k$
-- and if there is no such matching, then no coloring will pass the
test. 

We now formalize the idea behind color coding and then show how the
colorings can be computed in constant time. It turns out that one can
derandomize the computation of a coloring: instead of random colorings we
use sets of colorings such that for every set of $k$ vertices and
``desired'' colors for them, at least one coloring colors the vertices
as desired:

\begin{definition}[Universal Coloring Families]
  For natural numbers $n$, $k$, and $c$, an \emph{$(n,k,c)$-\penalty0universal
    coloring family} is a set $\Lambda$ of functions
  $\lambda\colon\{1,\dots,n\}\to\{1,\dots,c\}$ such that for every
  subset $S\subseteq\{1,\dots,n\}$ of size $|S| = k$ and for every 
  mapping $\mu\colon S\to\{1,\dots,c\}$ there is at least one function
  $\lambda \in \Lambda$ with $\forall s \in S\colon \mu(s) =
  \lambda(s)$. 
\end{definition}

The matching problem can be solved easily when we have
access to a $(n,2k,k)$-universal coloring family: If there is a
matching of size $k$, the family will contain some coloring that
colors the two endpoints of the first edge with color $1$, the
endpoints of the second edge with color $2$, and so on. Thus there is,
indeed, a matching of size $k$ in the graph if for at least one
coloring every color class contains an edge. Since we can easily check
in parallel for all colorings whether this is the case for one of them,
the complexity of $\PLang[\mathit k]{matching}$ hinges critically on
the complexity of computing the universal coloring family and the size
of this family. 
The next theorem shows that $(n,k,c)$-universal coloring families of
reasonable size can be
computed ``in constant time and work $f(k,c) \cdot n^{O(1)}$,'' which
implies that $\PLang[\mathit k]{matching} \in \Para\Class{AC}^0$ holds:

\begin{theorem}\label{theorem:universalColoring}
  {There is a \textsc{dlogtime}-uniform family
  $(C_{n,k,c})_{n,k,c\in\mathbb N}$ of $\Class{AC}$-circuits without
  inputs such that each $C_{n,k,c}$
  \begin{enumerate}
  \item outputs an $(n,k,c)$-universal coloring family (coded as a
    sequence of function tables), 
  \item has constant depth (independent of $n$, $k$, or $c$), and
  \item has size at most $O(n\log c \cdot c^{k^2} \cdot k^4\log^2 n)$.
  \end{enumerate}}
\end{theorem}

\begin{proof}[Sketch of Proof]
  The family of universal coloring functions we construct is based on
  the concept of $k$-perfect hash functions~\cite{FlumG2006}, that,
  after slight modifications, provide us with the desired coloring
  properties. The crucial part is to implement them using circuits
  that are \textsc{dlogtime}-uniform. However, we can achieve this,
  since the numbers $n$, $k$, and $c$ are encoded in unary and the
  operations required to compute the functions are only additions,
  multiplications, and modulo operations.
\end{proof}

Investigating a parameterized version of matching may seem a bit
strange at first sight -- matching is even known to be solvable in
randomized polylogarithmic parallel 
time. However, the exact parallel time complexity is still open in the
classical setting while from a parameterized perspective, we just saw
that the matching can be solved \emph{very} quickly in
parallel. Another problem that one would maybe not
expect to be studied in the parameterized setting,
but which will be useful in a number of situations, is
$\PLang{threshold}$. The inputs are a bitstring $b \in \{0,1\}^n$ and
a parameter~$t$. The question is whether there are at least $t$ many
$1$'s in $b$. Clearly, the unparameterized version is complete for
$\Class{TC}^0$, and using the fact that the problem lies in
$\Class{AC}^0$ for polylogarithmic thresholds~\cite{NewmanRW1990}
yields the fact that its parameterized version lies in
$\Para\Class{AC}^0$. However, this result requires profound
result of circuit complexity and is rather involved,
but using color coding we can give a very simple proof
of this fact:

\begin{lemma}\label{lem:threshold}
  $\PLang{threshold} \in \Para\Class{AC}^0$.
\end{lemma}

\subsection{Finding Embeddings of Graphs of Bounded Tree-Width and Depth}

A different way of looking at the matching problem is to see it as an
embedding problem: Instead of trying to find $k$ edges in a graph $G$
that have no endpoints in common, we can try to ``embed'' the graph~$H
= kK_2$, consisting of $k$ isolated edges, into~$G$. The advantage of
this different point of view is, of course, that it generalizes nicely:

\begin{problem}[{$\PLang{emb}(\mathcal H)$ for some
    class $\mathcal H$ of undirected graphs}]
  \begin{parameterizedproblem}
    \instance Two undirected graphs $H=(V_H, E_H) \in \mathcal H$ and
    $G=(V_G, E_G)$. 
    \parameter $H$
    \question Is there a injective homomorphism $\phi\colon V_{H}\to
    V_{G}$, that is, is $H$ isomorphic to a (not necessarily induced)
    subgraph of~$G$? 
  \end{parameterizedproblem}
\end{problem}

For arbitrary $\mathcal H$, the problem is easily be seen to be
$\Class W[1]$-hard by a reduction from $\PLang{clique}$. However, for 
restricted $\mathcal H$, the problem becomes fixed-parameter
tractable. The best results so far are by Chen and
M\"uller~\cite{ChenM2014B} who show that when $\mathcal H$ has bounded
tree-depth, $\PLang{emb}(\mathcal H)\in\Para\Class 
L$; when  $\mathcal H$ has bounded path-width, $\PLang{emb}(\mathcal
H)$ is the $\Para\Class L$-re\-duc\-tion closure 
of the distance problem in graphs, parameterized by the distance; and
when $\mathcal H$ has bounded tree-width, $\PLang{emb}(\mathcal H)$ is
the $\Para\Class L$-reduction closure 
of the embedding problem for trees, parameterized by the tree-size.
In contrast to these results, Amano showed for the unparameterized setting, in which we consider the size of $H$
to be a constant, that the problem can be solved in
$\Class{AC}^0$ with similar techniques~\cite{Amano:2010rm}.
We improve considerably on the first result of Chen and M\"uller
by proving that embeddings of graphs of bounded tree-depth can actually
be computed in $\Para\Class{AC}^0$. We complement their other results,
without improving them, by showing that for graphs of 
bounded tree-width (and, thereby, also for bounded path-width) the
embedding problem lies in $\Para\Class{AC}^{0\uparrow}$.  

In order to formulate our results, we first need to review the
definition of a tree-decomposition, see~\cite{FlumG2006} for a more detailed
introduction. A \emph{tree-decomposition} of a graph $H = (V,E)$ is a
tree $T$ together with a mapping $\iota$ from the nodes of $T$ to
subsets (called \emph{bags}) of~$V$ such that (1) for every edge
$\{u,v\} \in E$ there is some bag containing $u$ and $v$, that is,
there is some $x \in V$ with $\{u,v\} \subseteq \iota(x)$ and (2) for
every vertex $x \in V$ the set of nodes of~$T$ whose bags contain $x$
forms a connected subset of~$T$. The \emph{width} of tree-decomposition
is the size of its largest bag minus~$1$, its depth is the maximum of
the width and the depth of~$T$. Define $\operatorname{tw}(H)$ as the
minimum width any tree-decomposition of $H$ must have; define
$\operatorname{td}(H)$ similarly for the tree-depth.

\begin{theorem}\label{theorem:embtw}
  Given two graphs $H=(V_H,E_H)$ and $G=(V_G,E_G)$ together with a
  tree-decomposition $(T,\iota)$ of $H$. An embedding of $H$ into $G$ can be
  computed by an $\Class{AC}$-circuit of depth $O\big(\mathrm{depth}(T)\big)$
  and size $f(|V_H|)\cdot O\big(|V_G|^{\mathrm{width}(T)}\big)$, if
  such an embedding exists.
\end{theorem}

\begin{proof}[Sketch of Proof]
  Color the vertices of $H$ uniquely and compute a 
  $(|V_G|,|V_H|,|V_H|)$-universal coloring family. Starting from the
  leaves of the tree-decomposition, merge compatible partial
  homomorphisms for the vertices of the bags until we reach the root
  of the decomposition, and, thus, obtain a homomorphism for $H$. The
  number of iterative steps required for this equals the depth of the
  tree-decomposition.
\end{proof}
If $H$ is a parameter, we can compute a width- or depth-bounded
tree-decomposition $(T,\iota)$ of $H$ in a preprocessing step. This
implies the following corollaries:
\begin{corollary}\label{corollary:subgraphACk}
  Let $\mathcal H$ be the class of all graphs of tree-width at
  most~$d$ for some constant~$d$. Then
  $\PLang{emb}(\mathcal H)\in\Para\Class{AC}_{f(|H|)}\subseteq\Para\Class{AC}^{0\uparrow}$.
\end{corollary}

\begin{corollary}\label{corollary:subgraphAC}
  Let $\mathcal H$ be the class of all graphs of tree-depth at
  most~$d$ for some constant~$d$. Then
  $\PLang{emb}(\mathcal H)\in\Para\Class{AC}^0$.
\end{corollary}

We make two remarks at this point: First, one cannot generalize
Theorem~\ref{theorem:embtw} to clique-width since the embedding
problem for cliques, which have clique-width~1, is already hard for
$\Class W[1]$. Second, the theorem and the corollary also hold for
relational structures $\mathcal H$ and $\mathcal G$ and if we bound
the tree-width of $\mathcal H$'s Gaifman graph. Since paths have
tree-width~1, the complexity of one of the canonical problems for
color-coding~--~the $\PLang[\mathit k]{path}$ problem~--~can be
determined: $\PLang[\mathit k]{path}\in\Para\Class{AC}^{0\uparrow}$.
This allows us to give a short proof of
the following lemma on the complexity of the distance problem for
directed graphs where the distance is the parameter (one can also
prove this lemma directly quite easily): 

\begin{lemma}\label{lem:distance}
  $\PLang[\mathit d]{distance} \in \Para\Class{AC}_{f(d)}\subseteq\Para\Class{AC}^{0\uparrow}$.  
\end{lemma}

A known fact from circuit complexity states that  a polynomial-sized
$\Class{AC}$-circuit that decides whether a given graph $G$ contains a path of
length at most $d$ between to vertices $s$ and $t$ requires
depth $\Omega(\log\log d)$~\cite{Beame:1998uo}. This  implies
$\PLang[\mathit d]{distance}\not\in\Para\Class{AC}^{0}$.

\begin{corollary}\label{corollary:paraACvsParaACup}
$\Para\Class{AC}^0\subsetneq\Para\Class{AC}^{0\uparrow}$.
\end{corollary}

\subsection{First-Order Model Checking}

Our last result in this section on tools is an algorithmic
meta-theorem: We show that the model checking problem for first-order
logic on graphs of bounded degree lies in
$\Para\Class{AC}^{0\uparrow}$. We build strongly on a previous result by
Flum and Grohe~\cite{FlumG2002}, who showed that this model checking
problem lies in $\Para\Class L$, but differ in three regards: First,
we use color coding in our proof, which simplifies the
argument somewhat, second, we identify the parameterized distance
problem on bounded degree graphs as the only part of the computation
that is presumably not in $\Para\Class{AC}^0$, and, third, we observe
that the degree of the graphs can be made a parameter and need not be 
considered constant.

\begin{problem}[{$\PLang[\phi,\delta]{mc}(\Class{FO})$}]
  \begin{parameterizedproblem}
    \instance A logical structure $\mathcal A$ and a first-order formula $\phi$.
    \parameter The (size of) the formula $\phi$ and the maximum
    degree~$\delta$ of $\mathcal A$'s Gaifman graph.
    \question $\mathcal A \models \phi$?
  \end{parameterizedproblem}
\end{problem}

\begin{theorem}\label{thm:meta}
  $\PLang[\phi,\delta]{mc}(\Class{FO}) \in \Para\Class{AC}_{f(\phi+\delta)}\subseteq\Para\Class{AC}^{0\uparrow}$.
\end{theorem}

\begin{proof}[Sketch of Proof]
  By Gaifman's Theorem~\cite{Gaifman1982} we can rewrite the given
  formula as a formula~$\phi'$ in Gaifman normal form. Thus, what
  essentially remains is to check whether the structure (which we can
  interpret as a graph) contains $k$ disjoint ``balls'' of size
  bounded in the parameter (due to the maximum degree of the
  underlying Gaifman graph) that satisfy the subformulas in
  $\phi'$. To find these substructures, we make use of color coding
  and apply Lemma~\ref{lem:distance} to compute the corresponding
  connecting components. Finally, we only have to model check the
  resulting parameter-sized substructures.
\end{proof}

We conclude with the remark that the depth of the circuits constructed
in the above theorem just depends on the degree $\delta$ of the graph
and on the radius~$r$ of the balls,
which measure how ``local'' the formula~$\phi$ is. The smallest $r$
for which $\phi$ can be rewritten as in the proof is known as the
\emph{locality rank} $\operatorname{lr}(\phi)$ and the proof actually
shows that $\PLang[\phi,\delta]{mc}(\Class{FO})
\in \Para\Class{AC}_{O(\delta^{\operatorname{lr}(\phi)})}$.

\section{Fast Parallel Fixed-Parameter Algorithms for Natural Problems}
\label{section:applications}
  
The tools we have developed are now applied to a number of natural
parameterized problems found in the literature.

\subparagraph*{Packing Problems.}

We have already pointed out that the parameterized matching problem
can be seen as an embedding problem, where the objective is to embed
the graph $H = kK_2$, consisting of $k$ disjoint copies of a single
edge, into a graph~$G$. Embedding multiple disjoint copies of the same
graph into another graph is also known as ``packing''. Clearly,
instead of edges we can also pack other things as long as taking any
number of copies of these ``other things'' still has bounded
tree-depth. For instance, we can try to ``pack'' $k$ different
triangles into~$G$, that is, we can check whether there are $k$
vertex-disjoint triangles in~$G$. Unlike the matching problem,
triangle packing is known to be $\Class{NP}$-complete.

\begin{theorem}\label{thm:tripack}
  $\PLang{triangle-packing} \in \Para\Class{AC}^0$.
\end{theorem}

\begin{proof}
  Just observe that a graph $H$ consisting of any number of disjoint
  copies of a triangle has tree-depth~$3$. The claim follows from
  Corollary~\ref{corollary:subgraphAC}. 
\end{proof}

Indeed, for any fixed graph $H_0$ the packing problem
$\PLang{$H_0$-packing}$ lies in $\Para\Class{AC}^0$, where the 
question is whether we can find $k$ disjoint copies of $H_0$ in~$G$
and $k$ is the parameter:

\begin{theorem}\label{thm:h0packing}
  $\PLang{$H_0$-packing} \in \Para\Class{AC}^0$ for every fixed
  graph $H_0$.
\end{theorem}

Further variants arise when, instead of a single graph $H_0$, we are
given a whole multiset of graphs as inputs and we must find disjoint
copies of all of them in~$G$. Again, as long as there is a fixed bound
on the size of the graphs, the tree-depth of their disjoint union is
bounded and, hence, the packing problem lies in $\Para\Class{AC}^0$.

The complexity of packing problem changes when the to-be-packed graphs
no longer have constant size as in the following problem:

\begin{problem}[{$\PLang[\mathit{k,l}]{cycle-packing}$}]
  \begin{parameterizedproblem}
    \instance An undirected graph $G$ and two numbers $k$ and $l$.
    \parameter $k$ and $l$
    \question Are there $k$ vertex-disjoint cycles in $G$, each having length~$l$?
  \end{parameterizedproblem}
\end{problem}

The graph $H = k C_l$ consisting of $k$ copies of a cycle of
length~$l$ no longer has bounded tree-depth; it does have
tree-width~2, however. Thus, by Theorem~\ref{theorem:embtw} we get:

\begin{theorem}\label{thm:cycpack}
  $\PLang[\mathit{k,l}]{cycle-packing} \in \Para\Class{AC}_{f(k+l)}\subseteq\Para\Class{AC}^{0\uparrow}$.
\end{theorem}

The same result obviously also holds for
$\PLang[\mathit{k,l}]{path-packing}$ and it also holds for
$\PLang{forest-packing}$, where we are given a forest as input and  
the parameter is the total numbers of vertices in it.
We conclude with the remark that these ideas cannot be extended to
packing graphs whose tree-width is not bounded: Already embedding
cliques, let alone packing them, is $\Class W[1]$-hard.

\subparagraph*{Covering Problems.}

In \emph{covering problems} we must choose vertices in a graph (or
sometimes hypergraph) such that all ($\PLang{vertex-cover}$) or some
($\PLang{partial-vertex-cover}$) of the edges are ``covered,'' that is, they
intersect with the set of chosen vertices. The best-known covering
problem is undoubtedly $\PLang{vertex-cover}$, whose complexity has been 
scrutinized extensively in parameterized complexity
theory. We now prove $\PLang{vertex-cover} \in \Para\Class{AC}^0$; a
fact that nicely reflects on a theoretical basis the ``empirical''
observation that $\PLang{vertex-cover}$ is one of the ``easiest''
parameterized problems. The problem was one of the first
shown to lie in $\Para\Class P$, was then shown to lie in $\Para\Class
L$ by Cai  et~al.\ \cite{CaiCDF1997}, then in $\Para\Class{TC}^0$ by
Elberfeld and the last two authors~\cite{ElberfeldST2014}. 

\begin{theorem}\label{thm:vc in ac0}
  $\PLang{vertex-cover}\in\Para\Class{AC}^0$.
\end{theorem}

\emph{Partial} covering problems ask us not to cover all edges,
but only~$t$ of them: 

\begin{problem}[{$\PLang[\mathit{k,t}]{partial-vertex-cover}$}]
  \begin{parameterizedproblem}
    \instance An undirected graph $G=(V,E)$ and two numbers $k$ and $t$.
    \parameter $k$, $t$
    \question Is there a set $S\subseteq V$ of cardinality $|S|$ at
    most~$k$ such that the cardinality of $\bigl\{\{u,v\} \in E\mid u \in S \lor
    v\in S\bigr\}$ is at least~$t$?
  \end{parameterizedproblem}
\end{problem}
Another version is $\PLang[\mathit{t}]{exact-partial-vertex-cover}$,
where the size of $S$ is no longer restricted, but the cardinality of
$\bigl\{\{u,v\} \in E\mid u \in S \lor v\in S\bigr\}$ must be \emph{exactly}~$t$.

These problems, which are generally considered to be
harder than the plain vertex cover problem, lie in the class
$\Para\Class{AC}^{0\uparrow}$. Our proofs make an interesting use of
Theorem~\ref{thm:meta}. Recall that this ``meta-theorem'' states that all
first-order properties of graphs, parameterized by the first-order
property and the maximum degree of the graph, can be decided in
$\Para\Class{AC}^{0\uparrow}$. Covering properties can be expressed using
first-order formulas -- but we make \emph{no} requirement concerning
the degree of the input graph. The trick is to first reduce the inputs to
graphs of bounded degree and then apply the meta theorem. Such
a two-step approach is typically in advanced applications of
algorithmic meta-theorems.

\begin{theorem}\label{thm:pvc}
  $\PLang[\mathit{k,t}]{partial-vertex-cover} \in \Para\Class{AC}_{f(k+t)}\subseteq\Para\Class{AC}^{0\uparrow}$.
\end{theorem}

\begin{theorem}\label{thm:epvc}
  $\PLang[\mathit{t}]{exact-partial-vertex-cover} \in \Para\Class{AC}_{f(t)}\subseteq\Para\Class{AC}^{0\uparrow}$.
\end{theorem}

We conclude with the remark that the above results on finding vertex
coverings for graphs cannot easily be extended to hypergraphs since
for hypergraphs covering problems are typically hard for at least
$\Class W[1]$. 

\subparagraph*{Clustering Problems.}
\label{section:chromatic}

Clustering algorithms have a wide variety of applications, for
example in computational biology where we want to cluster genes and
proteins or process transcription data~\cite{BockerB2013}. A basic
clustering problem for graphs is the following:

\begin{problem}[{$\PLang[\mathit{k,\ell}]{cluster-editing}$}]
  \begin{parameterizedproblem}
    \instance An undirected graph $G=(V,E)$ and a numbers $\ell$
    and~$k$. 
    \parameter $\ell$, $k$
    \question Can we add and\,/\,or delete up to $k$~edges to or from
    $G$ such that the resulting graph consists of $\ell$ connected
    components, each of which is a clique?       
  \end{parameterizedproblem}
\end{problem}
A variant is $\PLang[\mathit k]{many-cluster-editing}$, where we just require
that the edited graph consists of cliques and do not prescribe the
number of clusters beforehand. This variant has been extensively
studied, most notably by Gramm et al.~\cite{GrammGHN2003} and
B\"ocker~\cite{Bocker2011} who showed its fixed-parameter
tractability. 
For the first version, algorithms based on color coding result in
reasonable running times, but where recently be outperformed by other
approaches~\cite{FominKPPV2011}. However, using a color coding
approach is useful when we consider parallel algorithms:

\begin{theorem}\label{theorem:clusterell}
  $\PLang[\mathit{k,\ell}]{cluster-editing}\in\Para\Class{AC}^0$. 
\end{theorem}

\begin{corollary}\label{corollary:manycluster}
  $\PLang[\mathit{k}]{many-cluster-editing}\in\Para\Class{AC}^0$.
\end{corollary}

We remark that if $\ell$ is not no longer considered a parameter in
cluster editing, the problem complexity increases only
moderately:

\begin{corollary}\label{corollary:cluster}
  $\PLang[\mathit{k}]{cluster-editing}\in\Para\Class{TC}^0$.
\end{corollary}

Theorem~\ref{theorem:clusterell} has another interesting corollary:
Let $\PLang[\mathit{k,p}]{complete-$p$-partite-editing}$ be the
problem of determining whether in a graph $G$ we can add and\,/\,or
remove up to $k$ edges such that the resulting graph is complete
$p$-partite, that is, its vertex set can be partitioned into exactly
$p$ non-empty sets such that there is an edge between two vertices if,
and only if, they belong to two different sets. Since the complement
of a complete $p$-partite graph is exactly a collection of $p$
cliques, we get the following corollary:

\begin{corollary}\label{corollary:completeppartite} \hfil
{  \begin{enumerate}
  \item 
    $\PLang[\mathit{k,p}]{complete-$p$-partite-editing}
    \in \Para\Class{AC}^0$.
  \item
    $\PLang[\mathit{k}]{complete-$p$-partite-editing} 
    \in \Para\Class{TC}^0$.
  \end{enumerate}}
\end{corollary}

Finally, instead of looking for just one complete $p$-partite graph,
we can look for several at the same time:

\begin{problem}[{$\PLang[\mathit{k,p}]{multipartite-cluster-editing}$}]
  \begin{parameterizedproblem}
    \instance An undirected graph $G=(V,E)$, a natural number $k$, and
    a sequence of natural numbers $p_1,p_2,\dots,p_\ell$.
    \parameter $k$, $p = p_1 + \cdots + p_\ell$
    \question Can we add or delete $k$ edges of $G$ such that the
    resulting graph consist of $d$ connected components $C_1$ to $C_\ell$
    such that each $C_i$ is a complete $p_i$-partite graph?
  \end{parameterizedproblem}
\end{problem}

\begin{theorem}\label{theorem:multiclustering} \hfil
  {\begin{enumerate}
  \item 
    $\PLang[\mathit{k,p}]{multipartite-cluster-editing}\in\Para\Class{AC}^0$
  \item
    $\PLang[\mathit{k,\ell}]{multipartite-cluster-editing}\in\Para\Class{TC}^0$.
  \end{enumerate}}
\end{theorem}

\subparagraph*{Graph Separation Problems.}
\label{section:conquer}

Graph separation problems are problems where we ask to separate a set
of $\ell$~vertices from the remaining graph by deleting at most $k$
other vertices. They play a key role in many real-world network
applications like finding communities or isolating dangerous
vertices. While this problem is well-known to be $\Class{NP}$-complete
in the unparameterized setting and $\Class{W}[1]$-hard in the
parameterized setting for parameters $k$, $\ell$, and $k+\ell$, the
complexity of the problem changes dramatically if we require the
separated set of vertices to be connected:

\begin{problem}[{$\PLang[\mathit k,\ell]{cutting-$\ell$-connected-vertices}$}]
  \begin{parameterizedproblem}

    \instance An undirected graph $G=(V,E)$ and two natural numbers
    $k$ and $\ell$.

    \parameter $k$, $\ell$

    \question Is there a partitioning of $V$ into three sets $X$, $S$,
    and $Y$ with $|X|=\ell$ and $|S|\leq k$ such that $X$ is connected
    and for all $\{x,y\} \in E$ with $x\in X$ we have $y\not\in Y$?

  \end{parameterizedproblem}
\end{problem}
Marx~\cite{Marx2006} showed that this problem is fixed-parameter
tractable; Fomin, Golovach, and Korhonen~\cite{FominGK2013} studied a
similar version, namely $\PLang[\mathit
k,\ell]{cutting-at-most-$\ell$-vertices}$, in which the set $X$ is not
required to be connected and may be of size \emph{at most} $\ell$,
i.\,e., $1<|X|\leq\ell$, and for which Fomin et al.\ gave an
$\Class{FPT}$-algorithm based on color coding. The main idea is to
colorize the given graph with two colors such that the vertices of the
set $X$ get colored with the first color and the vertices in $S$
get the second color. Thus, we only have to find the
solution within the vertices of the first color.  
This algorithm can be implemented in $\Para\Class{AC}^{0\uparrow}$ and, moreover,
works for $\PLang[\mathit
k,\ell]{cutting-$\ell$-connected-vertices}$ as well.

\begin{theorem}\label{theorem:cutting} \hfil 
{  \begin{enumerate}
  \item 
    $\PLang[\mathit k,\ell]{cutting-$\ell$-connected-vertices}\in\Para\Class{AC}_{f(\ell)}\subseteq\Para\Class{AC}^{0\uparrow}$.
  \item 
    $\PLang[\mathit k,\ell]{cutting-at-most-$\ell$-vertices}\in\Para\Class{AC}_{f(\ell)}\subseteq\Para\Class{AC}^{0\uparrow}$.
  \end{enumerate}}
\end{theorem}

We conclude with the remark that both problems can also be solved with
algorithms similar to 
the ones presented above if we consider the \emph{terminal versions}
of these problems~\cite{FominGK2013}, i.\,e., there is a special
terminal vertex~$t$ which has to be part of $X$. For this, we have to
modify the above algorithms to consider only blue components that
contain~$t$.

\section{Conclusion}

We have seen that many natural parameterized problems can be solved in
constant parallel time or in parallel time depending only on the
parameters while doing only ``\textsc{fpt} work.'' We stress that our
results are of a theoretical nature and do not directly give practical
parallel implementations for the problems presented; but they show
that such implementations are possible \emph{in principle} for them. 
The core technique used in all proofs (at least indirectly) was
\emph{color coding,} which can be done in constant time and which is
already used in practice. 

This paper did not address lower bounds. While for $\Para\Class{AC}^0$
this is not problematic since this class lies at 
the bottom of almost any hierarchy of parameterized classes,
some problems in $\Para\Class{AC}^{0\uparrow}$ might well ``fall down'' to
$\Para\Class{AC}^0$. Here we only know a explicit lower bound for the
distance problem, which does not lie in
$\Para\Class{AC}^0$. Establishing lower bounds for other problems in $\Para\Class{AC}^{0\uparrow}$
is therefore a reasonable research goal.

\bibliographystyle{plain}
\bibliography{main}

\clearpage
\appendix

\section{Technical Appendix: Proofs}

For the readers convenience, the claims of the proofs given in this
appendix are repeated before the proofs.

\reclaimtheorem{theorem:universalColoring}

\begin{proof}
  Define 
  \begin{align*}
    \lambda_{p,a}(x) &= (a\cdot x\bmod p)\bmod k^2, \\
    \Lambda'_{n,k} &= \bigl\{\,\lambda_{p,a}\mid \text{$p$ is a
                              prime with $p<k^2\log n$ and 
                              $a\in\{0,\dots,p-1\}$}\bigr\},
    \\
    \Lambda_{n,k,c} &= \bigl\{\omega\circ\lambda_{p,a}\mid 
                      \omega\colon\{0,\dots,k^2-1\}\to\{1,\dots,c\},\;
          p<k^2\log n,\;
          a\in\{0,\dots,p-1\}\bigr\}.
  \end{align*}
  It is well-known that $\Lambda'_{n,k}$ is a family of $k$-perfect
  hash functions, i.\,e., for every subset $S\subseteq\{1,\dots,n\}$
  with $|S|=k$ it contains a function that is injective on $S$, see
  \cite{FlumG2006}. Therefore, given a subset $S$ and a function $\mu
  \colon S \to \{1,\dots, c\}$, some member of $\lambda_{p,a} \in
  \Lambda'_{n,k}$ will map the members of~$S$ injectively to a subset
  $S'$ of $\{0,\dots,k^2-1\}$ and, then, some function $\omega\colon
  \{0,\dots,k^2-1\} \to \{1,\dots, c\}$ will map $S'$ in such a way
  that $\omega \circ \lambda_{p,a}$ equals $\mu$ on $S$. Consequently
  the set $\Lambda_{n,k,c}$ is an $(n,k,c)$-universal coloring
  family. Notice that we use all $p< k^2\log n$ in the definition of
  $\Lambda_{n,k,c}$ and, thus, including the prime numbers only
  indirectly.  The sizes of the two sets can be bounded by
  $|\Lambda'_{n,k}|\leq (k^2\cdot\log n)^2$ and $|\Lambda_{n,k,c}|\leq
  \smash{c^{k^2}}\cdot (k^2\cdot\log n)^2 = \smash{c^{k^2}}k^4\log^2
  n$. Each function in $\Lambda_{n,k,c}$ can clearly be encoded in $n
  \log_2 c$ bits.
  
  For the construction of circuits~$C_{n,k,c}$ observe that they
  have no inputs and must just output $\Lambda_{n,k,c}$ in a fixed
  encoding. Thus, we can, in principle, hardwire the complete
  output of $C_{n,k,c}$ into a depth-0 circuit. The tricky part is, of
  course, arguing that the circuit family is
  \textsc{dlogtime}-uniform. However, this is surprisingly simple:
  Having a look at the definition of $\Lambda_{n,k,c}$, we see that 
  computing the $i$th bit of its encoding only involves simple
  computations consisting of additions, multiplications, and modulo
  operations on $i$ and the numbers $n$, $k$, and $c$. Now,
  \textsc{dlogtime}-uniformity means that we have time \emph{logarithmic}
  in the   \emph{unary} encodings of these numbers and hence
  \emph{polynomial} time in   their \emph{binary} encodings. Since
  addition, multiplication, and modulo   are clearly polynomial-time
  computable, we get the claim. 
\end{proof}

\reclaimlemma{lem:threshold}

\begin{proof}
  On input of a bitstring $b$ of length~$n$ and a number~$t$, use
  Theorem~\ref{theorem:universalColoring} to compute an $(n, t,
  t)$-universal coloring family. Now, if $b$ contains at least $t$
  many $1$'s, then there is a coloring of the positions of~$b$ such
  that each color class contains at least one~$1$. Thus, it suffice to
  test in parallel for all colorings whether this is the case. 
\end{proof}

\reclaimtheorem{theorem:embtw}

\begin{proof}
  Let $N$ denote the set of nodes of~$T$. For a node $n$
  of $T$, let $T_n$ be the subtree of $T$ rooted at~$n$ and  
  let $N_n$ be the set of its vertices.
  Color $G$ by a $(|V_G|,|V_H|, |V_H|)$-universal coloring family using
  Theorem~\ref{theorem:universalColoring} and
  test all members of this family in parallel. To simplify the notation,
  let us identify the $|V_H|$ colors in each coloring with the
  vertices of $V_H$.
  
  Let us call a subset $X \subseteq V_G$ \emph{colorful} if all
  vertices of $X$ have a different color; let $Y \subseteq V_H$ be the
  set of these colors and let $\mu_X \colon Y \to X$ map each color
  $y$ in $Y$ to the vertex $x$ in $X$ having this color. Note that if
  there is an embedding $\phi$ of~$H$ into $G$, for at least one
  coloring there is a $\mu_X=\phi$ such that $X = \phi(V_H) \subseteq  V_G$ is
  the image of~$V_H$. 
  
  Let us call a colorful subset $I \subseteq V_G$ of size at most
  $\mathrm{width}(T)+1$ \emph{good for a node $n$ of $T$} if there is a colorful superset
  $J \supseteq I$ of vertices of~$V_G$ such that $\mu_J$ is an
  injective homomorphism $\mu_J \colon \bigcup_{m \in N_n} \iota(m)
  \to J$. In other words, $I$ can be extended to a solution of the
  embedding problem for the tree rooted at $n$.
  
  Clearly, since the size of the $\iota(n)$'s and the $I$'s are
  restricted by the $\mathrm{width}(T)+1$, for a leaf $n$ of $T$ we can
  decide whether a subset is good for some node of~$T$ using only a
  constant number of $\Class{AC}$-layers of width bounded by a
  function in $\mathrm{width}(T)$. Also observe that the number
  of subsets of $V_G$ of size at most $\mathrm{width}(T)+1$ is bounded
  by $|V_G|^{\mathrm{width}(T)+1}$ and,
  thus, we can consider all of them in parallel in each layer of an
  $\Class{AC}$-circuit. 
  
  We must now show that we can decide, using only as many layers as
  the depth of~$T$, whether $H$ has an embedding in~$G$. In
  a first layer, we first determine for each leaf $n$ of~$T$ the set
  of all good sets $I$ for~$n$. In the inductive step, consider a node
  $n$ such that for all  its children we have already determined which
  sets are good for 
  them. Let $I$ be a colorful set for which we must determine whether
  it is good for~$n$. We claim that this is the case when two
  conditions are met: 
  \begin{enumerate}
  \item The set $I$ is ``a correct embedding itself,'' meaning that
    $\mu_I$ is an injective homomorphism $\mu_I \colon \iota(n) \to
    I$. 
  \item The children of $n$ ``can be made consistent with $I$,''
    meaning that for each child $c$ of $n$ in~$T$, there is a colorful
    set $I_c \subseteq V_G$ that is good for $c$ and $\mu_I$ and
    $\mu_{I_c}$ are identical on $\iota(n) \cap \iota(c)$.
  \end{enumerate}
  To see that these tests suffice in order to test whether $I$ is good
  for $n$, just observe that by the definition of a tree-decomposition, all bags that contain a particular vertex $h$ of $V_H$
  must form a connected subset of~$T$. Our second condition ensures
  that when a given vertex $g \in V_G$ has been picked as the image
  of~$h$ in some~$I$, the same vertex must have been picked in all
  children and, thus, a partial homomorphism on $I$ can be extended to
  a homomorphism on the vertices in the whole tree rooted at~$n$.
  
  To conclude the proof, we just observe that the two tests can
  clearly be implemented with a constant number of
  $\Class{AC}$-layers. 
\end{proof}

\reclaimlemma{lem:distance}
\begin{proof}
  Determining whether there is a path of length $d$ from $s$ to $t$ in
  $G$ is the same as asking whether the path $P_d$ of length $d$ can
  be embedded into a directed graph $G$  with the start and  end of
  the path marked appropriately so that they must be mapped to $s$
  and~$t$, respectively. Since paths have tree-width~1,
  Theorem~\ref{theorem:embtw} and the second remark following it give
  the claim.
\end{proof}

\reclaimtheorem{thm:meta}
\begin{proof}
  The first part of our proof is identical to the one given by Flum
  and Grohe in~\cite{FlumG2002}: Let $\phi$ be a formula given as
  input. For simplicity of presentation, we assume that the structure
  $\mathcal A$ is actually an undirected graph $G = (V,E)$ of maximum
  degree $\delta$. Let $d(a,b)$ denote the distance of two vertices
  in~$G$ and let $N_r(a) = \{\,b \in V \mid d(a,b) \le r\,\}$ be the ball
  around $a$ of radius $r$ in~$G$. Let $G[N_r(a)]$ denote the subgraph
  of $G$ induced on $N_r(a)$. By Gaifman's Theorem~\cite{Gaifman1982} 
  we can rewrite $\phi$ as a Boolean combination of formulas of the
  following form:
  \begin{align*}
    \exists x_1 \cdots \exists x_k \Bigl(\textstyle \bigwedge_{i\neq j}
    \psi_{\mathrm{dist}>2r}(x_i,x_j) \land \bigwedge_i \psi(x_i)\Bigr)       
  \end{align*}
  where $\psi_{\mathrm{dist}>2r}(x_i,x_j)$ is a standard formula
  expressing that $d(x_i,x_j) > 2r$ and $\psi$ is \emph{$r$-local,}
  meaning that for all $a \in V$ we have $G \models \psi(a) \iff
  G[N_r(a)] \models \psi(a)$.  What remains to be done is to determine
  whether there are $k$ 
  vertices $a_1$ to $a_k$ in $G$ such that the balls $N_r(a_i)$ do not
  intersect and $G[N_r(a_i)] \models \psi(a_i)$ holds for them.

  At this point, we digress from the line of argument of Flum and
  Grohe, who now give a slightly involved space-efficient algorithm
  for determining the existence of such $a_i$ without having to write
  them down (which is not possible in parameterized logarithmic
  space). Instead, we use color 
  coding at this point: Introduce colors $1$ to $k$. Since the maximum
  degree~$\delta$ is a parameter and $r$ depends only on a parameter,
  the maximum size $M$ of any $N_r(a)$ is bounded by the
  parameter. This means that there is an $(|V|, M k, k+1)$-universal
  coloring family such that for the vertices $a_i$ from above at least
  one coloring has the following property: All vertices in
  $N_r(a_i)$ have the same  color~$i$. This means that each
  $N_r(a_i)$ is contained in a monochromatic connected component
  of~$G$ having color~$i$. 

  It remains to test whether for each color~$i$ there is a
  vertex~$a_i$ such that $N_r(a_i)$ has color~$i$ and $G[N_r(a_i)]
  \models \psi(a_i)$ holds. For this, let some candidate $a_i$ be
  given. We need to determine for a given vertex~$b$ whether $d(a_i,b)
  \le r$ where the distance is computed in the subgraph of $G$ induced
  by the vertices of color~$i$.   In other words, we need to solve the problem
  $\PLang[\mathit{d,\delta}]{undirected-distance}$, which is
  parameterized over the distance $d$ and the maximum
  degree~$\delta$ and which can be solved in $\Para\Class{AC}_{f(r)}$
  by Lemma~\ref{lem:distance}.\footnote{Flum and Grohe argue that the
    undirected distance problem can be solved in space $f(r,\delta) +
    O(\log n)$, since we just need $r \log_2 \delta$ bits to describe
    a path of length $r$ starting at a vertex~$a$ and can iterate over
    all possible paths with that many bits.}
  Once the set $N_r(a)$ of vertices reachable from a vertex~$a$ in at
  most $r$ steps has been determined, we can create an isomorphic
  copy of $G[N_r(a)]$ consisting just of an $|N_r(a)| \times |N_r(a)|$
  adjacency matrix in $\Para\Class{AC}^0$: Number the vertices of~$G$
  in some manner (for instance, in the order they appear in the
  input), which also induces an ordering on the vertices of
  $N_r(a)$. The entry in row~$i$ and column~$j$ of the matrix is a $1$
  if the $i$th and the $j$th vertex in $N_r(a)$ are connected by an
  edge in $E_G$. Determining  which vertex is the $i$th vertex of 
  $N_r(a)$ can be done by a $\Para\Class{AC}^0$ circuit by
  Lemma~\ref{lem:threshold}.  

  Given the adjacency matrix of $G[N_r(a)]$, we can clearly decide in
  $\Para\Class{AC}^0$ whether $G[N_r(a)] \models \phi$, since the size
  of $G[N_r(a)]$ depends only on the original input parameters. 
\end{proof}

\reclaimtheorem{thm:vc in ac0}

\begin{proof}
  Let us reiterate the steps of the well-known Buss kernelization:
  Let $G$ be an input graph and let $k$ be the size of
  the sought vertex cover. First, we can determine, in parallel, all
  vertices $v$ that have degree at least $k+1$ and, as observed by
  Buss, all of these vertices must be part of any vertex cover of size
  at most $k$. Remove these vertices from $G$
  in parallel and then remove all isolated vertices. Buss' second
  observation is that if the  remaining graph has more than $k(k+1)$
  vertices, no vertex cover of size $k$ exists. Thus, we get a
  quadratic problem kernel.

  Elberfeld et~al.~\cite{ElberfeldST2014} observe that the essential
  parallel steps of the algorithm are the following: (1)~Check whether
  the degree of 
  a vertex is at least $k$. (2)~Checking whether there are at most $k$
  such vertices. (3)~Checking whether there are at most $k(k+1)$
  vertices that are not only connected to the high-degree
  vertices. (4)~Computing the subgraph induced by these $k(k+1)$
  vertices.

  While Elberfeld et~al.\ conclude at this point that the computation
  can be implemented by $\Para\Class{TC}^0$~circuits (``we just have to
  count''), Lemma~\ref{lem:threshold} shows that the computation can
  be implemented using a $\Para\Class{AC}^0$~circuit: All counting
  involves thresholds depending only of the parameter. 
\end{proof}

\reclaimtheorem{thm:pvc}

\begin{proof}
  On input of a graph $G$, first test whether there is a vertex of
  degree at least~$t$. If so, we can accept since this vertex alone
  already constitutes the desired cover. Otherwise, we know that the
  \emph{graph has a maximum degree bounded by the parameter} and we
  can apply Theorem~\ref{thm:meta} to the following first-order
  formula, which depends only on $k$ and~$t$:
  \begin{align*}
    \underbrace{\exists x_1 \cdots \exists x_k}_{\text{the size-$k$
        cover}} \underbrace{\exists a_1 \exists b_1 \cdots 
    \exists a_t \exists b_t}_{\text{the $t$ covered edges}} \Bigl(
    \phi_{\operatorname{dist}}(a_1,b_1,\dots,a_t,b_t) \land \textstyle
    \bigwedge_{i=1}^t \bigl(E(a_i,b_i) \land \bigvee_{j=1}^k a_i=x_j\bigr)\Bigr).
  \end{align*}
  Here, $\phi_{\operatorname{dist}}$ is a standard formula expressing
  that $\{a_1,b_1\}$, \dots, $\{a_t,b_t\}$ are distinct sets. 
\end{proof}

\reclaimtheorem{thm:epvc}

\begin{proof}
  We again wish to apply Theorem~\ref{thm:meta}, but now the
  preprocessing step is a bit more complicated: Vertices of degree
  higher than $t$ no longer constitute a solution, indeed, these
  vertices \emph{cannot} be selected as part of a solution. However,
  we also cannot simply remove them as we did in the proof of
  Theorem~\ref{thm:vc in ac0} since parts of their neighbors might be chosen
  and the edges attached to them are then part of the $t$ covered
  edges. The trick is to replace all vertices~$v$ of degree $d > t$ 
  by $d$ new vertices and to add an edge from each of the former
  neighbors of $v$ to exactly one of these $d$ new vertices. (It is
  not difficult to implement these replacement steps in constant depth
  by also adding some unnecessary isolated vertices.) Let us color all
  ``new'' vertices red. 

  Once the graph has been preprocessed, it will once more have degree
  bounded in the parameter and we can apply Theorem~\ref{thm:meta} to
  a formula stating ``there exist $k$ vertices and  $t$ distinct
  edges such that the $k$ vertices are not red, the $t$ edges always
  have one  endpoint among the $k$ vertices and all edges of the graph
  having an endpoint among the $k$ vertices are among the $t$ edges.'' 
\end{proof}

\reclaimtheorem{theorem:clusterell}

\begin{proof}
  Let $G = (V,E)$ and $k$ be given as input. For the moment, assume
  that $G$ \emph{can} be clustered after $k$ edge
  modifications and let $C = \{C_1,\dots,C_\ell\}$ be a solution, that
  is, a partitioning of $V$ such that $R = \{\,\{u,v\} \mid u \in C_i, v
  \in C_j, i\neq j\,\} \cap E$ (these edges need to be
  removed) and $A = \{\,\{u,v\} \mid u,v \in C_i\,\} \setminus
  E$ (these edges need to be added) together have size is at
  most~$k$. Define $M = \bigcup R \cup \bigcup A$ as the set of all
  vertices attached to edges that need to be 
  modified. Let us call a cluster $C_i$ \emph{partly modified} if
  $C_i \not\subseteq M$, and \emph{completely modified} if $C_i
  \subseteq M$.

  Our objective is to determine the clusters~$C_i$ without knowing
  them. Towards this aim, we apply color coding for an
  $(n,2k+\ell,2)$-universal coloring family with the colors blue and
  orange. If a clustering~$C$ exsists, at least one coloring has the
  following two properties:  
  \begin{enumerate}
  \item All vertices in $M$ are colored blue. (Hence, all completely
    modified clusters will be completely blue.)
  \item In each partly modified cluster $D$, at least one vertex
    $D \setminus M$ is colored orange. Let $d$ be the smallest
    such vertex with respect to the ordering of the vertices in the
    input. 
  \end{enumerate}

  To identify the partly modified clusters, we consider only orange
  vertices. Since all vertices in $M$ are colored blue, edges incident
  to orange vertices will not change. Adjacent orange vertices will
  therefore belong to the same cluster and non-adjacent ones will
  belong to different clusters. Hence, we can 
  identify the vertex~$d$ in each partly modified cluster: It is an
  orange vertex that is not adjacent to a smaller 
  orange vertex. Observe that in $G$ the vertex~$d$ has the following
  property:
  \begin{quote}
  ($*$) All vertices connected to $d$ form a partly modified
    cluster $D$.
  \end{quote}

  Thus, the orange vertices that are not adjacent to smaller orange
  vertices induce a set of partly modified clusters. We count their
  number and count the total number $m_1$ of modifications needed to
  form them. Since $m_1 \le
  k$ must hold, we can compute $m_1$ in constant parallel time by
  Lemma~\ref{lem:threshold}. 
  
  It remains to consider the number of modifications needed to form
  the completely modified clusters. However, the total number of
  vertices in these clusters is at most~$2k$; so after conceptually
  removing all vertices that are part of partly modified clusters, at
  most $2k$ vertices may remain. We can compute the subgraph induced
  by these vertices in constant time 
  (using the same argument as in the proof of Theorem~\ref{thm:meta}
  for the construction of the subgraph induced on $N_r(a)$) and then
  solve this kernel in constant time, yielding a minimum number $m_2$
  of modifications needed to create the completely covered
  clusters. We accept when $m_1 + m_2 \le k$ and the number of partly
  modified clusters and completely modified clusters is~$\ell$.
\end{proof}

\reclaimcorollary{corollary:manycluster}

\begin{proof}
  The argument is similar to the one from
  Theorem~\ref{theorem:clusterell}, but we first apply a preprocessing
  to find components of~$G$ that are already cliques. Since the number
  of sought clusters (cliques) is not limited, no optimal solution will
  ever modify edges adjacent to vertices in such a clique and, thus,
  we can conceptually remove them from the input. To identify these
  vertices, let us call a vertex \emph{cliquish} if all its
  neighbors are pairwise connected in~$G$; a property that we can
  easily test in constant time. The set $X$ of cliquish vertices
  contains exactly all vertices of clusters already present in~$G$.

  The argument now continues as in Theorem~\ref{theorem:clusterell},
  only we (1) completely ignore the vertices in $X$ and (2) look for
  \emph{up to $\ell = 2k$} partly or completely modified clusters
  rather than exactly $\ell$ such clusters. 
\end{proof}

\reclaimcorollary{corollary:cluster}

\begin{proof}
  Our argument starts as in the proof of
  Corollary~\ref{corollary:manycluster}: We identify components that are
  already cliques. However, this time, we only add such a
  component to~$X$ if (1) its size is larger than $k$ because, then,
  it cannot be part of any editing or if (2) there are $2k$ such
  components of the same size $s \le k$ earlier in the
  input\footnote{A component $A$ is ``earlier in the input than a
    component~$B$'' if there is a vertex in $A$ whose position in the
    input is before all vertices in~$B$.}, because we can apply any
  necessary modifications to these $2k$ earlier components. Now, we
  count the number $x$ of clusters in $X$ (this is the only place
  where we need a $\Class{TC}^0$ circuit). If $\ell' := \ell - x >
  2k + k(k+1)/2$, we know the coloring cannot lead to a solution. Otherwise, we
  ask whether the graph $G$ without~$X$ together with the numbers
  $\ell'$ and~$k$ is an instance of
  $\PLang[\mathit{k,\ell}]{cluster-editing}$. 
\end{proof}

\reclaimtheorem{theorem:multiclustering}

\begin{proof}
  The start of our proof is identical to the one of 
  Theorem~\ref{theorem:clusterell} and we use the same
  terminology. The first difference concerns the property ($*$) from the
  proof: When all clusters are cliques, an orange vertex in such a
  clique immediately identifies all vertices in it, namely as the set
  of its neighbors. For $p$-partite graphs, this is more difficult and
  we use a new definition of equivalence: Let us call two orange
  vertices \emph{equivalent} if they are adjacent in~$G$ or if they
  have the same neighborhood in~$G$. On orange vertices, this is,
  indeed, an equivalence relation and some vertices $d$ will have the
  property that they are minimal with respect to this relation. For
  such vertices, we make a new observation:
  \begin{quote}
  ($**$) The partly modified cluster containing $d$ consists of
    $d$, all orange vertices equivalent to~$d$, all blue vertices
    connected to any of these vertices, and possibly some additional
    vertices from~$M$. 
  \end{quote}
  
  The ``additional vertices from~$M$'' arise for instance in a
  bipartite cluster when one shore contains only blue vertices and the
  other contains some blue and some orange vertices. Then the blue
  vertices of the second shore do not have any orange neighbors. Let
  $Y$ be the set of all vertices that are ``definitely identified'' by
  rule ($**$), that is, the set of all orange vertices and together with
  its neighborhood. If $|V| - |Y| > 2k$ we can stop, since only
  vertices from $M$ may be missing from~$Y$. We can also stop when the
  number of identified partly modified clusters is more
  than~$\ell$.

  We can now identify all partly modified clusters~$D$ -- except that
  some of the vertices in~$M$ may still be lacking --, but we do not
  yet know which one is $C_1$, which one is $C_2$, and so on. We try 
  out, in parallel, all possible injective mappings 
  from the identified clusters to the set of indices
  $\{1,\dots,\ell\}$ together with all possible ways of mapping the at
  most $2k$ vertices in $V \setminus Y$ to $\{1,\dots,\ell\}$. Each
  pair of mappings determines a possible clustering
  $\{C_1,\dots,C_\ell\}$ and  we can now (1) compute the number of
  edge removals that are necessary to remove all edges between
  clusters and (2) use Corollary~\ref{corollary:completeppartite} to
  determine the minimal number of editing operations necessary to make
  the $i$th cluster $C_i$ a complete $p_i$-partite graph in
  $\Para\Class{AC}^0$ or $\Para\Class{TC}^0$, depending on whether the
  values of the $p_i$ are parameters or not. We accept when the total
  number of modifications is at most~$k$.
\end{proof}

\reclaimtheorem{theorem:cutting}
\begin{proof}
  We begin with the first item. For
  this, we make use of a family of $(n,k+\ell,2)$-universal coloring
  functions. If $G$ contains a set~$X$ of $\ell$ vertices that
  can be separated from the remaining vertices by removing a set $S$
  of at most $k$ vertices, then the family of coloring functions
  contains a coloring such that the vertices of $X$ are colored with
  the first color, say blue, and the vertices of the set $S$ are
  colored with the second color, say orange. Hence, we iterate over
  these colorings, and for each coloring we try to find $X$ and $S$ by
  searching connected components of blue vertices in the
  graph. For this, we iterate over all vertices $x$ of the graph and
  each time check whether it is part of a set $X$ with the desired
  properties: 
  
  We can find out whether a vertex $y$ has distance at most $d$
  from $x$ in the blue subgraph in $\Para\Class{AC}_{f(d)}$ by
  Lemma~\ref{lem:distance}. If there is some $y$ at distance $\ell+1$,
  we know that the component containing $x$ is too large and we can
  stop. Otherwise, we can identify all vertices $y$ reachable from $x$
  inside the blue component in $\Para\Class{AC}_{f(\ell)}$. If the
  number of such vertices in $\ell$ (we can test this even in constant
  depth using Lemma~\ref{lem:threshold}), test whether the number of
  orange vertices connected to any such $y$ is at most~$k$ (again,
  this test can be done in constant time).

  To prove the second item, we proceed in a
  similar way, but instead of searching for connected components of
  size~$\ell$, we search for a blue component of size at most $\ell$
  that has at most $k$ orange neighbors. 
\end{proof}

\end{document}